%% file: main.tex
\documentclass[12pt]{article}
\pdfoutput=1

\input{usepackages}
\input{commands}

\input{mytheorems}

\input{comments}

\title{Free Fermion Distributions Are Hard to Learn}
\author[1]{Alexander Nietner\thanks{Email: \href{mailto:a.nietner@fu-berlin.de}{a.nietner@fu-berlin.de}}}
\affil[1]{\fu}
\date{}

\begin{document}
\maketitle

\input{content/abstract.tex}

\input{content/introduction.tex}

\input{content/preliminaries.tex}

\input{content/overview.tex}

\input{content/conclusion.tex}

\section*{Acknowledgements}
\input{content/acknowledge}

\printendnotes

\printbibliography

\end{document}

%% file: usepackages.tex
\usepackage[utf8]{inputenc}
\usepackage{lipsum}

\usepackage{amsmath, amssymb,amsfonts,amsthm,mathtools}
\usepackage{xspace,graphicx,relsize,bm,xcolor}
\usepackage{soul} 

\usepackage{parskip}  

\usepackage[sb]{libertine} 
\usepackage[T1]{fontenc}
\usepackage{dsfont}  
\usepackage{textcomp}
\usepackage[varqu,varl]{zi4}
\usepackage[amsthm]{libertinust1math} 
\usepackage[cal=stix,scr=boondoxo,bb=txof]{mathalpha} 

\usepackage{xurl} 
\usepackage[
    backend=biber,
    style=alphabetic,
    sorting=anyt,
    minalphanames=3,
    maxalphanames=3,
    maxnames=99,
    backref=true
    ]{biblatex}
\DefineBibliographyStrings{english}{%
  backrefpage = {page},
  backrefpages = {pages},
}

\usepackage{enotez}

\usepackage{hyperref}
\usepackage[margin=1.75cm]{geometry}
\definecolor{linkcol}{rgb}{0.0,0.55,0.7}
\definecolor{citecol}{rgb}{0.0, 0.6, 0.45}
\definecolor{urlcol}{rgb}{0.7, 0.0, 0.55}
\hypersetup{
	colorlinks,
	linkcolor={linkcol},
	citecolor={citecol},
	urlcolor={urlcol}
}

\usepackage{subcaption}
\usepackage{mleftright}
\usepackage{hyperref}
\usepackage{multirow}
\usepackage{physics}  

\usepackage{cleveref}
\usepackage{authblk}  

%% file: commands.tex
\def\01{\{0,1\}}
\newcommand{\mc}[1]{\mathcal{#1}}
\newcommand{\mb}[1]{\mathds{#1}}

\DeclareMathOperator*{\Ex}{\mathbf{E}}
\let\Pr\relax
\DeclareMathOperator*{\Pr}{\mathbf{Pr}}

\newcommand{\eps}{\epsilon}

\newcommand{\A}{\ensuremath{\mathcal{A}}}
\newcommand{\D}{\ensuremath{\mathcal{D}}}
\newcommand{\G}{\ensuremath{\mathcal{G}}}
\newcommand{\M}{\ensuremath{\mathcal{M}}}

\newcommand{\gen}{{\mathtt{Gen}}}
\renewcommand{\eval}{\mathtt{Eval}}

\newcommand{\parity}{\mathrm{par}}
\newcommand{\poly}{\mathrm{poly}}

\newcommand{\samp}{\mathtt{Sample}}

\newcommand{\stat}{\mathtt{Stat}}

\newcommand{\swap}{\mathtt{SWAP}}
\newcommand{\fswap}{\mathtt{FSWAP}}

\newcommand{\Pf}{\mathrm{Pf}}
\newcommand{\tv}{\mathrm{d}_\mathrm{TV}}

\newcommand{\lr}[1]{\left(#1\right)}
\newcommand{\lrq}[1]{\left[#1\right]}
\newcommand{\lrb}[1]{\left\{#1\right\}}

\newcommand{\fu}{Dahlem Center for Complex Quantum Systems, Freie Universit\"{a}t Berlin, 
14195 Berlin, Germany}

%% file: mytheorems.tex
\newtheoremstyle{mydefinitionsty}
{10pt}
{10pt}
{}
{}
{}
{}
{.5em}
{\textbf{\thmname{#1}~\thmnumber{#2}:  }\thmnote{(#3)}}
\theoremstyle{mydefinitionsty}
\newtheorem{definition}{Definition}

\newtheoremstyle{myproblemsty}
{10pt}
{10pt}
{}
{}
{}
{}
{.5em}
{\textbf{\thmname{#1}~\thmnumber{#2}:  }\thmnote{(#3)}\newline}
\theoremstyle{myproblemsty}
\newtheorem{problem}{Problem}

\newtheoremstyle{mythmsty}
{10pt}
{10pt}
{\itshape}
{}
{}
{}
{.5em}
{\textbf{\thmname{#1}~\thmnumber{#2}:  }\thmnote{(#3)}}
\theoremstyle{mythmsty}
\newtheorem{inftheorem}{Informal Theorem}\Crefname{inftheorem}{Informal Theorem}{Informal Theorems}

\newtheorem{lemma}{Lemma}
\newtheorem{corollary}{Corollary}

\newtheorem{conjecture}{Conjecture}

%% file: comments.tex
\definecolor{alexcolor}{rgb}{0.0, 0.47, 0.75}   

\definecolor{questioncolor}{rgb}{0.36, 0.54, 0.66}

%% file: content/abstract.tex
\begin{abstract}
    Free fermions are some of the best studied quantum systems. 
    However, little is known about the complexity of learning free-fermion distributions: the output distributions of free fermion states after measuring in the occupation number basis. 
    In this work we establish the 
    hardness of this task in the particle number non-preserving case, partially resolving an open question by \cite{aaronsonEfficientTomographyNonInteracting2023}. 
    In particular, we give an information theoretic hardness result for the general task of learning from 
    expectation values and, in the more general case when the algorithm is given access to individual samples, we give a computational hardness result based on the LPN assumption for learning the probability density function. The latter immediately implies the computational hardness of learning a free fermion state with matching output distribution.
\end{abstract}

%% file: content/introduction.tex
\section{Introduction}

In this work we investigate the learnability of the output distributions of free fermions after measuring in the occupation number basis.
We refer to those distributions simply as \emph{free fermion distributions}.
This learning task can be seen as the analogue to tomography of free fermion states from a classical perspective. 
The learnability of free fermion distributions was first analyzed in the particle preserving case in an earlier version of \cite{aaronsonEfficientTomographyNonInteracting2023}. 
The learning algorithm by Aaronson and Grewal later turned out to be incomplete as detailed by the authors \cite{aaronsonEfficientTomographyNonInteracting2023} Section 1.2.2, leaving open the question about the learnability of free fermion distributions, which they pose as an open question.

Due to Wicks theorem it is known that tomography of free fermion states can be done efficiently: 
the $n$-point correlators are fully determined by the $2$-point correlators.
In particular, a free fermion state is robustly  determined via its covariance matrix, which in turn is determined by the two point correlators \cite{gluza2018, swingle2019} which can be directly estimated from local expectation values \cite{aaronsonEfficientTomographyNonInteracting2023, nietnerUnifyingQuantumStatistical2023}.  

Here we establish the hardness of learning free fermion distributions, 
which by the previous paragraph immediately yields a strict separation between learning from local observables against learning from samples in the occupation number basis. 
We focus on the Gaussian case, where  the global particle number is not fixed. 
This is, our results do not imply hardness of learning free fermion distributions with a fixed particle number as originally considered by \cite{aaronsonEfficientTomographyNonInteracting2023}. 
While this work brings the hardness-of-learning transition closer to the latter family of distributions, our technique can not easily be applied to fixed particle number free fermion distributions.

As established by \cite{terhal_classical_2002}, free fermions are closely related to local match gate circuits as introduced in \cite{valiant_quantum_2002}. 
In particular, nearest neighbor one dimensional match gate circuits are equivalent to nearest neighbor one dimensional free fermion circuits. 
Thus, the free fermion formalism mildly generalizes the match gate formalism since the former can also be classically simulated when the nearest neighbor criterion is dropped. 
This however can also be emulated via match gates using a network of fermionic $\swap$ operations (i.e. $\fswap$) which comes at an at most linear overhead in the system size. 
Thus, the difference is fairly mild.
We will use the match gate formalism for our technical formulae and give the free fermion interpretation in words.


\subsection{Our Contribution}

We state our results asymptotically in terms of the scaling parameter \(n\) being the number of modes, or likewise the number of qubits when working in the match-gate setting.
Our first  result is concerned with learning algorithms that make use of statistical averages, only.

\begin{inftheorem}[C.f. \Cref{cor:stat-hard} and \Cref{cor:stat-hard-ff}]\label{inf:stat}
Free fermion distributions are exponentially hard to learn from empirical expectation values.
\end{inftheorem}

Interestingly, this result is of information theoretic nature: the hardness is due to the fact that any algorithm requires the values of exponentially many expectation values in order to determine the underlying distribution.
This is in stark contrast to free fermion tomography: 
By Wick's theorem all $k$-point functions can be decomposed into combinations of $2$-point functions. 
Thus, the quadratically many $2$-point functions, which can be estimated from expectation values, suffice to determine the corresponding quantum state\endnotemark.  

\endnotetext{We note that it is actually easy to see that any algorithm for learning free fermion distributions using only two point functions must have some blind spots: The set of free fermion distributions contains all tensor products of even parity distributions (c.f. \Cref{lem:M-s-in-M}). The even parity distribution, on the other hand, has uniform marginals. Thus, there exist pairs of distributions that can not be distinguished by the knowledge of all two point functions.}

Since Wicks theorem is at the core of the free fermion formalism, the lack of its applicability for distribution learning hints at the possibility of a stronger statement. The next statement goes into this direction as that it allows the learner 
to access the unknown distribution
directly on the level of individual samples (instead of only empirical expectation values). 
It relies on the learning parities with noise (LPN) assumption, a standard assumption in cryptography \cite{regevOnLattices2009, pietrzakCryptography2012}.

\begin{inftheorem}[C.f. \Cref{cor:samp-hard} and \Cref{cor:samp-hard-ff}]\label{inf:sample}
Free fermion distributions are hard to learn from samples under the learning parity with noise assumption.
\end{inftheorem}

In other words, learning free fermion distributions takes at least as long as solving LPN.
The best known algorithm for the latter task is the BKW algorithm due to Blum, Kalai and Wasserman (or variations thereof) and runs in time \(2^{O\left(n/\log(n)\right)}\) \cite{blumNoiseTolerantLearningParity2000a}.

\Cref{inf:stat,inf:sample} have interesting consequences on the power of the associated data-access models implicitly assumed. 
While \Cref{inf:stat} assumes the learning algorithm to access the unknown free fermion distribution only by coarse expectation values modelled by statistical queries (see \Cref{d:oracles}), \Cref{inf:sample} assumes the algorithm to have access to individual samples. 
In contrast, as argued above, by Wick's theorem the free fermion state and hence its output distribution can be learned efficiently from empirical expectation values of quantum observables, the latter being modelled by quantum statistical queries, see \cite{nietnerUnifyingQuantumStatistical2023} and similarly \cite{aaronsonEfficientTomographyNonInteracting2023} for the particle number conserving case. 
In particular, due to Theorem 7.6 in \cite {nietnerUnifyingQuantumStatistical2023} free fermion states and hence their corresponding Born distribution can be learned from \(O(n^2)\) quantum statistical queries. 

\begin{corollary}[Information Separation]
    Learning the output distribution of a free fermion state requires exponentially more empirical expectation values from the output distribution compared to empirical expectation values from the underlying state. 
    In particular, \(c=2^{\Omega(\sqrt{q})}=2^{\Omega(n)}=q 2^{\Omega(n)}\) empirical classical expectation values are required in terms of the number of empirical quantum expectation values \(q\).
\end{corollary}

We can similarly interpret \Cref{inf:sample} in terms of an almost exponential
speedup in learning from quantum versus classical data, assuming the optimality of BKW.

\begin{corollary}[Computational Separation]
    There is a super-polynomial computational speedup in learning free fermion distributions from quantum data from the underlying free fermion states versus from the classical data sampled from the distribution.
    In particular, assuming the optimality of BKW, then 
    for any \(\gamma<1\) there is a constant \(0<\gamma^{\prime}<\gamma\), such that the time complexity of learning from classical data \(t_c\) scales as \(t_c=2^{\Omega(t_q^{\gamma^{\prime}})}=t_q 2^{\Omega(n^\gamma)}\)  where \(t_q\) is the time complexity of learning from quantum data.
\end{corollary}

%% file: content/preliminaries.tex
\section{Preliminaries}\label{preliminaries}

These preliminaries are sectioned into two parts. The first part is about distribution learning
and the second part is about free fermion, respectively match gate distributions. 

\subsection{Distribution Learning}

\noindent We denote by $\mathcal{F}_n$ the set of Boolean functions from $\lrb{0,1}^n$ to $\lrb{0,1}$ and by $\D_n$ the set of probability distributions over $\lrb{0,1}^n$. The parity of a bitstring $x$ is denoted by $\abs{x}$ and is given by $\abs{x}=\sum_i x_i\mod2$. A subset $\D\subset\D_n$ is referred to as 
a distribution class. 
For two probability distributions $P,Q:\01^n \rightarrow [0,1]$, we denote by ${\tv(P,Q):= \frac{1}{2} \sum_{x\in \01^n} |P(x)-Q(x)|}$ the total variation distance between them.
For an $n$-qubit quantum state \(\rho\) we refer to the \(n\)-bit distribution \(P(x)=\mathrm{tr}[\ketbra{x}{x}\rho]\) as the output distribution of \(\rho\).
Access to distributions is formalized by an oracle with a specific operational structure. 
In this work, we consider the sample and the statistical query oracle (see also \cite{Kearns:1994:LDD:195058.195155, kearns_efficient_1998, feldman2017statistical}).

\begin{definition}[Distribution oracles]\label{d:oracles}
Given $P\in\mathcal{D}_n$ and some $\tau \in (0,1)$. We define:
\begin{enumerate}
    \item The sample oracle $\samp(P)$ as the oracle which, when queried, provides a sample $x\sim P$. 
    \item The statistical query oracle $\stat_{\tau}(P)$ as the oracle which, when queried with a function $\phi:\{0,1\}^n\rightarrow [-1,1]$, responds with some $v$ such that $|\Ex_{x\sim P}[\phi(x)] - v| \leq \tau$. 
\end{enumerate}
We say that an oracle presents a distribution.
\end{definition}

A learning algorithm has to output the distribution in the sense of some representation.
In this work we focus on two popular representations: generators and evaluators.
These formalize generative and density models, respectively, and are defined as follows. 

\begin{definition}[Representations of a distribution]\label{d:gen}\label{d:eval}
    Given $P\in\D_n$, we say that
    \begin{enumerate}
        \item a probabilistic (or quantum) algorithm $\gen(P)$ is a generator for $P$ if it produces samples according to $x\sim P$.
        \item an algorithm $\eval(P):\lrb{0,1}^n\to[0,1]$ is an evaluator for $P$ if, on input $x\in\lrb{0,1}^n$ it outputs $\eval(P)[x]=P(x)$. 
    \end{enumerate}
    We say that a generator (evaluator) represents a distribution.
\end{definition}

To formalize the notion of \textit{distribution learning},
we use the framework introduced in Ref.~\cite{Kearns:1994:LDD:195058.195155}.
This definition is analogous to the definition of probably-approximately correct function learning, in that it introduces parameters $\eps$ and $\delta$ to quantify approximation error and probability of success. 

\begin{problem}[$(\eps,\delta)$-distribution-learning]\label{prob:eps-del-learning}
Let $\eps,\delta\in(0,1)$ and let $\D$ be a distribution class. Let $\mathcal{O}$ be a distribution oracle. The following task is called $(\eps,\delta)$-distribution-learning $\D$ from $\mc O$ with respect to a  generator (evaluator): Given access to oracle  $\mc O(P)$ for any unknown $P\in\D$, output with probability at least $1-\delta$ a generator (evaluator) of a distribution $Q$ such that $\tv(P,Q)<\eps$. 
\end{problem}

\begin{definition}[Efficiently learnable distribution classes] Let $\mathcal{D}$ be a distribution class, and let $\mathcal{O}$ be a distribution oracle. We say that $\mathcal{D}$ is computationally (query) efficiently learnable from $\mathcal{O}$ with respect to a generator/evaluator, if there exists an algorithm $\mathcal{A}$ which for all $(\epsilon,\delta) \in (0,1)$ solves the problem of $(\epsilon,\delta)$-distribution learning $\mathcal{D}$ from $\mathcal{O}$ with respect to a generator/evaluator, using $O(\mathrm{poly}(n,1/\epsilon,1/\delta))$ computational steps (oracle queries). 
Moreover, the generator/evaluator which is learned must itself be an efficient algorithm, i.e. terminates after $O(\poly(n))$ many steps.
\end{definition}

As we are most often concerned with computational efficiency and with the sample oracle, we often omit these qualifiers in this case, and simply say ``$\mathcal{D}$ is efficiently learnable". If a distribution class is not efficiently learnable, then we say it is hard to learn.

\subsection{Free Fermion and Match Gate Distributions}

Free fermion states are states that can be written as gaussian states in the fermionic creation operators
\begin{align}
    \ket{\psi}=\mc N\exp\lr{\sum_{ij}G_{ij}c^\dagger_ic^\dagger_j}\ket{0}\,,
\end{align}
where $G$ is an anti-symmetric generating matrix, $\mc N$ a normalization term and the $c_i^\dagger$ are the fermionic creation operators while $\ket{0}$ is the vacuum state. The amplitudes in the mode basis are then given by
\begin{align}
    \braket{x}{\psi}=\mc N\Pf\lr{G\vert_x}\,.
\end{align}
Here, $\Pf$ is the Pfaffian and $G\vert_x$ denotes the anti symmetric sub-matrix of $G$ indicated by the $1$'s in the bitstring~$x$. Equivalently, these states can be created by circuits of (particle number non-preserving) free fermion unitaries. As shown in \cite{terhal_classical_2002},  on the level of qubits these correspond to nearest neighbor match gate circuits on a line.

\begin{definition}[Match Gate]\label{def:match-gate}
We say a $2$-qubit unitary of the form
\begin{align}
    U=e^{i\phi}
    \begin{pmatrix}
        W_{11}&0&0&W_{12}\\
        0&Q_{11}&Q_{12}&0\\
        0&Q_{21}&Q_{22}&0\\
        W_{21}&0&0&W_{22}
    \end{pmatrix}
\end{align}
is a match gate  if $W,Q\in\mathrm{SU}(2)$ and $\phi\in[0,2\pi)$. We denote by $\G$ the set of all match gates. Quantum circuits composed of match gates are referred to as match gate circuits. 
\end{definition}

Let us define two particularly important match gates in the context of this work.

\begin{definition}
The fermionic swap gate $\fswap$ is defined as 
\begin{align}
    \fswap=
    \begin{pmatrix}
        1&0&0&0\\
        0&0&1&0\\
        0&1&0&0\\
        0&0&0&-1
    \end{pmatrix}\,.
\end{align}
Moreover we define
\begin{align}
    U_{X}(t)=e^{itX\otimes X}=
    \begin{pmatrix}
        \cos(t/2)&0&0&i\sin(t/2)\\
        0&\cos(t/2)&i\sin(t/2)&0\\
        0&i\sin(t/2)&\cos(t/2)&0\\
        i\sin(t/2)&0&0&\cos(t/2)
    \end{pmatrix}=\cos(t/2)\mb 1_4+i\sin(t/2) X\otimes X\,,
\end{align}
where $X$ is the usual Pauli-$X$ gate.
\end{definition}

The $\fswap$ gate is identical to the $\swap$ gate except for the $-1$ phase in the $((11),(11))$ entry. This phase corresponds to the fermionic statistics when swapping two occupied modes, thus the name fermionic swap. Any non-local two body interaction in the fermionic formalism can be represented by a circuit of local intereactions in the match gate formalism at the price of a linear increase in the depth of the circuit.

Finally, we define the actual class of distributions we are interested in: the Born distributions corresponding to free fermion states when measured in the occupation number basis. This is equivalent to measuring match gate states in the computational basis. Thus we define the following.

\begin{definition}[Match Gate Distributions]\label{def:match-gate-distribution}
We define the class of (nearest neighbour one dimensional) match gate distributions on $n$ qubits at depth $d$ to contain all $n$-bit distributions of the form
\begin{align}
    P(x) = \abs{\mel{x}{U}{0^n}}^2
\end{align}
where $U$ is any nearest neighbour one dimensional match gate circuit on $n$ qubits of depth $d$. 
The set of all such distributions is denoted by $\M(n,d)$ or simply $\M$. 
\end{definition}

The set of one dimensional nearest neighbor free fermion distributions is identical to the set of match gate distributinos at the corresponding depth and particle number. Similarly, we identify the set of non-local free fermion distributions at a given depth with the set of local match gate distributions at the same depth without the $\fswap$ gates taken into account.

%% file: content/overview.tex
\section{Results}

In this section we will show the hardness of learning $\M$ both, from statistical queries as well as from samples. To this end we will embed parities and then noisy parities. The hardness of learning then follows in analogy to \cite{hinsche_single_2022}. We begin with the following definition.

\begin{definition}[Parity]\label{def:parity}
For any $s\in\01^n$ denote by $\chi_s:\01^n\rightarrow\01$ the party function  defined as
\begin{align}
    \chi_s(x)=x\cdot s\,,
\end{align}
where the scalar product is taken in $\mb F_2^n$. Moreover the corresponding {parity distribution} $D_s\in\D_{n+1}$ is defined as
\begin{align}
    D_s(x,y)=
    \begin{cases}
        2^{-n}\,,&\chi_s(x)=y\\
        0\,&\text{else,}
    \end{cases}
\end{align}
with $(x,y)\in\01^n\times\01$ and denote by $\D_\chi$ the set of all such distributions.
For any noise rate $\eta\in[0,1]$ we define the {noisy parity distribution} as 
\begin{align}
    D_s^\eta(x,y)=
    \begin{cases}
        (1-\eta)\cdot2^{-n}\,,&\chi_s(x)=y\\
        \eta\cdot2^{-n}\,&\text{else.}
    \end{cases}
\end{align}
Denote the set of all such distributions by $\D_\chi^\eta$.
\end{definition}

In this work we will focus on two distribution classes derived from those. First consider the ``fermionized'' version of $\D_\chi$ (note that this  ``fermionization'' is not unique). Here, $\abs{x}$ denotes the parity of $x$ and addition of bits refers to the modulo two addition as in $\mb{F}_2$ (such as $\abs{x}+y$ for bit string $x$ and bit $y$).

\begin{definition}[Fermionized Parity]\label{def:fermionized-parity}
For any $s\in\01^n$ define the
fermionized parity function $\xi_s:\01^n\rightarrow\01^2$ as
\begin{align}
    \xi_s(x)=
    \begin{pmatrix}
        \chi_s(x)\\
        \chi_s(x)+\abs{x}
    \end{pmatrix}\,.
\end{align}
Next, we define the
{fermionized parity distribution} $M_s\in\D_{n+2}$ as
\begin{align}
    M_s(x,y,z)=
    \begin{cases}
        2^{-n}\,,&\chi_s(x)=y\,\text{and}\,\abs{x}+y=z\\
        0\,&\text{else.}
    \end{cases}
    =
    \begin{cases}
        2^{-n}\,,&\xi_s(x)=(y,z)\\
        0\,&\text{else.}
    \end{cases}
\end{align}
The set of all such distributions is denoted by $\M_\chi$
\end{definition}

The interpretation of $\xi_s$ and $M_s$ is as follows: Any fermionic state has a parity constraint. Thus, we can not encode the distributions $D_s$ corresponding to a random example oracle for the parity function $\chi_s$ directly into a state on $n+1$ qubits. Rather, we need an auxiliary qubit that takes care of the global parity constraint. This is done via \Cref{def:fermionized-parity}.

The following lemma states an equivalence between $M_s$ and $D_s$.

\begin{lemma}\label{lem:gen-eval-reduction}
An evaluator (generator) for $M_s$ for some unknown $s$ implies an evaluator (generator) for $D_s$ and vice versa.
\end{lemma}

\begin{proof}
In case of an evaluator: 
Assume an evaluator for $M_s$,
let $(x,y)\in\01^n\times\01$. Define $z=\abs{x}+y$. Then we can simulate an evaluator for $D_s$ as $\eval_{D_s}(x,y)=\eval_{M_s}(x,y,z)$. 
Conversely, assuming an evaluator for $D_s$, let $(x,y,z)\in\01^n\times\01^2$. We first check whether $\abs{x}+y=z$. If this is the case we compute $\eval_{M_s}(x,y,z)$ from $\eval_{D_s}(x,y)$. Else we return $0$.

In case of a generator we can simulate $\gen_{D_s}$ by simply discarding the last bit of the output of $\gen_{M_s}$. Conversely, we can simulate $\gen_{M_s}$ from $\gen_{D_s}$ by sampling a bit string $(x,y)$ and appending $\abs{x}+y$.
\end{proof}

\subsection{Statistical Query Lower Bound}

With this in mind the following lemma the leads to hardness of learning $\M_\chi$ in the statistical query setting.

\begin{lemma}[Statistical Query Reduction]\label{lem:stat-reduction}
Any statistical query to $\xi_s$ ($M_s$) with tolerance $\tau$ can be simulated via two statistical queries to $\chi_s$ ($D_s$) with tolerance $\frac\tau2$.
\end{lemma}

\begin{proof}
Let $\phi:\01^{n}\times\01^2\rightarrow[-1,1]$ be a function which we want to query.
Then $\phi$ can be decomposed in three parts $\phi=\phi_e+\phi_o+\phi_0$ where:
\begin{itemize}
    \item $\phi_e(x,y,z)=\vcentcolon \hat\phi_e(x,y)$ is non-zero only when $(x,y,z)$ is such that $\abs{x}+y=z=0$.
    \item $\phi_o(x,y,z)=\vcentcolon \hat\phi_o(x,y)$ is non-zero only when $(x,y,z)$ is such that $\abs{x}+y=z=1$.
    \item $\phi_0$ is such that $\Ex_{M_s}[\phi_0]=0$ for all $s$.
\end{itemize}
Thus
\begin{align}
    \Ex_{(x,y,z)\sim M_s}\lrq{\phi(x,z,y)}=
    \Ex_{(x,z)\sim D_s}\lrq{\hat\phi_e(x,y)}+
    \Ex_{(x,z)\sim D_s}\lrq{\hat\phi_o(x,y)}\,,
\end{align}
which implies that two statistical queries to $D_s$ with tolerance $\frac\tau2$ suffice to simulate a statistical query to $M_s$ with tolerance $\tau$.
\end{proof}

Thus, any statistical query algorithm for learning $\M_\chi$ implies a statistical query algorithm for learning $\D_\chi$ with twice the query complexity. The latter is known to be exponentially hard \cite{blum_weakly_1994} thus the former must be hard, too. This  gives the following lemma.

\begin{lemma}\label{lem:stat-m-chi-hard}
Let $\epsilon<1/4$ and $\delta<1/2-2^{-3n}$. $(\epsilon,\delta)$-learning $\M_\chi$ requires at least $\Omega(2^{n/3-2})$ statistical queries with tolerance $\tau=\Omega(2^{-n/3+1})$.
\end{lemma}

Note, that we did not clarify the representation with respect to which the learner has to learn. 
This is, because the proof is of information theoretic nature and thus holds for any representation. 

\begin{proof}
The proof is analogous to the proof of Theorem 4 in \cite{hinsche_single_2022}.
Assume an algorithm $\A$ which is able to $(\epsilon,\delta)$-learn $\mc M_\chi$ for $\epsilon,\delta<1/2$ using $q$ many statistical queries with tolerance $\tau$ and with respect to some representation.
Using \Cref{lem:stat-reduction} we obtain an algorithm $\mc B$ that given statistical query access with tolerance $\tau/2$ to some $D_s\in\D_\chi\subset\D_{n+1}$ learns the corresponding representation for some $M\in\D_{n+2}$ such that $\tv(M,M_s)<\epsilon$.  
We can use this representation of $M$ to compute the corresponding underlying parity function without any further statistical queries.
To see this, note that for any $r\neq t$ it holds $\tv(M_r,M_t)\geq1/2$. 
Thus, $s$ is the unique bit string for which $M_s$ is at most $1/2-\epsilon<1/4$ far from $M$, which can be found by brute force without any query to $D_s$.

In particular, algorithm $\mc B$ makes $2q$ statistical queries with tolerance $\tau'=\tau/2$ to $\D_\chi$ in order to learn the correct parity function.
Hence, for $\tau=\Omega(2^{-n/3+1})$ it must hold that $q=\Omega(2^{n/3-2})$ {\cite[Theorem 12]{blum_weakly_1994}}.
\end{proof}

With the following observation we can thus conclude that learning the corresponding match gate distributions is hard.

\begin{lemma}\label{lem:M-s-in-M}
\begin{align}
    \M_\chi\subset\M(n,O(n))
\end{align}
\end{lemma}

\begin{proof}
Let $\parity_k$ denote the uniform distribution on even parity bit strings on $k$ bits
\begin{align}
    \parity_k(x)=
    \begin{cases}
        2^{-n+1}\,,\quad &\abs{x}=0\\
        0\,,&\text{else}\,.
    \end{cases}
\end{align}
Every distribution $M_s$ can be written as
\begin{align}\label{eq:pi-parity}
    M_s=\Pi\cdot(\parity_m\otimes\parity_{n+2-m})\,,
\end{align}
where $\Pi$ is a permutation on $n+2$ bits consisting of swaps of bits only and thus can be written as a linear depth network of $\swap$ gates, and where $m-1\leq n$ is the number of $1$'s in $s$. To see this we note that the parity constraint  $z=\abs{x}+y$ is equivalent to $z=\abs{x\vert_{\neg s}}$ where $x\vert_{\neg s}$ denotes the sub string of $x$ labelled by the $0$'s in $s$. 

Next we note that $\parity_k$ can be realized as the Born distribution of a depth $2$ match gate circuit on $k$ qubits. The corresponding circuit consists only of $U_X(\pi/2)$ gates. To see this we note that $U_X(\pi/2)$ applied to $\ket{ij}$ flips with probability $0.5$  both bits to the state $\ket{\neg i \neg j}$ and leaves both in the $\ket{ij}$ state with probability $0.5$ (up to a phase). 
Thus, the depth two brick work circuit composed of those gates applied to the all zero state creates every even parity bit string with an equal probability.

Finally, we observe that applying an $\fswap$ before the measurement has the same effect as applying a $\swap$ on the post measurement distribution. In particular, 
\begin{align}
    D^{\otimes 2}\circ (\mathtt{FSWAP}\otimes \mathtt{FSWAP}^\dagger)=\mathtt{SWAP}\circ D^{\otimes 2}\,,
\end{align}
where $\fswap\otimes\fswap^\dagger$ is the quantum channel corresponding to $\fswap$ and $D$ is the local computational basis measurement operator which maps the quantum state to its diagonal vector
\begin{align}
    D(\ketbra{i}{j})=\delta_{ij}\ket{i}\,.
\end{align}

Combining this with \Cref{eq:pi-parity} we find that any $M_s$ can be implemented by a depth $O(n)$ match gate circuit.
\end{proof}

Together with Lemma 7 in analogy to Theorem 4, both from  \cite{hinsche_single_2022}, we obtain the following corollary.

\begin{corollary}[Local Free Fermion Distributions: Formal version of \Cref{inf:stat}]\label{cor:stat-hard}
There is no efficient algorithm for learning $\M$ at depth $d=\omega(\log(n))$ from inverse  polynomial accurate queries. Equivalently, learning $\M$ at depth $\Omega(n)$ from statistical queries with tolerance $\tau=\Omega(2^{-n/3+1})$ requires $\Omega(2^{n/3-1})$ many queries.
\end{corollary}

From the perspective of free fermions it is also natural to consider the same question without a locality constraint. On the level of match gates this amounts as counting the $\fswap$ as a free resource. Comparing with the proof of \Cref{lem:M-s-in-M} we then find the following corollary which is probably the strongest formulation of our result.

\begin{corollary}[Non-local Free Fermion Distributions: Formal version of \Cref{inf:stat}]\label{cor:stat-hard-ff}
Learning non-local free fermion distributions at constant depth $d\geq2$ from statistical queries with tolerance $\tau=\Omega(2^{-n/3+1})$ requires $\Omega(2^{n/3-1})$ many queries.
\end{corollary}

\subsection{Hardness for Learning from Samples}

Let us now consider the sample oracle. To this end we are going to embed the learning parities with noise (LPN) problem into $\M$. The LPN problem is defined in the context of \textit{probably approximately correct} (PAC) learning.

\begin{definition}[PAC learning]
    Let $\epsilon,\delta>0$, let $\mc C$ be a class of boolean functions $f:\01^n\to\01$ and let $D$ be a distribution over $\01^n$. We say an algorithm $\A$ efficiently $(\epsilon,\delta)$-PAC learns $\mc C$ with respect to $D$ if, for any $f\in\mc C$, the algorithm receives $N$ samples $(x_1,f(x_1)),\dots,(x_N,f(x_N))$ with $x\sim D$, and, with probability $1-\delta$ returns a boolean function $h$, such that
    \begin{align}
        \Pr_{x\sim D}\lrq{f(x)\neq h(x)}<\epsilon\,.,
    \end{align}
    where the run time of $\A$ (and thus also the number of samples $N$) is bounded by $O(\poly(n,1/\epsilon,1/\delta))$
\end{definition}

\begin{conjecture}[Learning Parities With Noise]\label{con:LPN}
There is a constant $0<\eta<1/2$ such that there is no efficient algorithm for learning parity functions under the uniform distribution in the PAC model with classification noise rate $\eta$.
\end{conjecture}

Thus, LPN states that there is a constant $\eta$ such that it is hard to learn $s$ when given access to $D_s^\eta$. In \cite{Kearns:1994:LDD:195058.195155} it is shown how this implies hardness of learning $D_s^\eta$ with respect to an evaluator in the sense of distribution learning. In particular, learning $\D_\chi^\eta$ is at least as hard as LPN.

We will now embed $\D_\chi^\eta$ into $\M$ in order to obtain the corresponding hardness result for match gates. Again we start by defining the ``fermionized'' LPN distribution.

\begin{definition}[Fermionized Noisy Parity]\label{def:fermionized-noisy-parity}
For any $s\in\01^n$ and noise rate $0\leq\eta\leq1$ we define the fermionized noisy parity distribution $M_s^\eta$ over $\01^{n+2}$ as
\begin{align}
    M_s^\eta(x,y,z) =
    \begin{cases}
        (1-\eta)\cdot2^{-n}\,,&\chi_s(x)=y\,\text{and}\,\abs{x}+y=z\\
        \eta\cdot2^{-n}\,,&\chi_s(x)=\neg y\,\text{and}\,\abs{x}+y=z\\
        0\,&\text{else.}
    \end{cases}
    =
    \begin{cases}
        (1-\eta)\cdot2^{-n}\,,&\xi_s(x)=(y,z)\\
        \eta\cdot2^{-n}\,,&\xi_s(x)=(\neg y,\neg z)\\
        0\,&\text{else.}
    \end{cases}
\end{align}
The set of all such distributions is denoted by $\M_\chi^\eta$.
\end{definition}

Again, the structure is such that, by definition, the distribution $M_s^\eta$ is supported only on even parity strings and, on those, encodes $D_s^\eta$. This is made clear in the following lemma similar to \Cref{lem:gen-eval-reduction}.

\begin{lemma}\label{lem:sample-reduction}
The generators $\gen(D_s^\eta)$  and $\gen(M_s^\eta)$ simulate each other, and hence the corresponding sample oracles simulate each other. Similarly,  $\eval(D_s^\eta)$  and $\eval(M_s^\eta)$ simulate each other.
\end{lemma}

\begin{proof}
Let $(x,y)\sim D_s^\eta$. Then $(x,y,\abs{x}+y)\sim M_s^\eta$. Contrarily, let $(x,y,z)\sim M_s^\eta$. Then $(x,y)\sim D_s^\eta$.

The statement on evaluators follows in both directions from the defining equation
$D_s^\eta(x,y)=M_s^\eta(x,y,\abs{x}+y)$ together with $M_s^\eta(x,y,\neg(\abs{x}+y))=0$.
\end{proof}

We are now able to show hardness of learning $\M_\chi^\eta$ with respect to an evaluator.

\begin{lemma}
Under the LPN assumption there is no efficient algorithm for learning $\M_\chi^\eta$.
\end{lemma}

\begin{proof}
The proof is analogous to the proof of Theorem 16 in \cite{Kearns:1994:LDD:195058.195155}.
Let $\A$ be an algorithm that efficiently $(\epsilon,\delta)$-learns $\M_\chi^\eta$. We will now construct an algorithm $\mc B$ that solves LPN. We first use \Cref{lem:sample-reduction} in order to efficiently transform the noisy parity oracle to a sample oracle for some unknown $M_s^\eta$. 
We then run $\A$ in order to obtain, with probability $1-\delta$ an evaluator for a distribution $M$ with $\tv(M,M_s^\eta)<\epsilon$. 
For a uniform random $x$ we can, with probability $1-\epsilon$ over the choice of $x$,  compute $\chi_s(x)$ by checking whether $M(x,0,\abs{x})$ or $M(x,1,\abs{x}+1)$ is larger and return accordingly. Thus we can $(\epsilon,\delta)$-PAC learn $\chi_s$ with respect to the uniform distribution. We conclude, by \Cref{con:LPN} that $\A$ must be inefficient.
\end{proof}

We now show that the fermionized noisy parity distribution actually is contained in $\M$.

\begin{lemma}
\begin{align}
    \M_\chi^\eta\subset\M(n,O(n))\,.
\end{align}
\end{lemma}

\begin{proof}
Let $\ket{\psi_s}$ be some match gate state that has $M_s$ as its Born distribution (c.f. \Cref{lem:M-s-in-M}). Then, applying $U_X(t)$ to the $y,z$ register of $\ket{\psi_s}$ with $\sin^2(t/2)=\eta$ results in a state $\ket{\psi'_s}$ with $M_s^\eta$ as corresponding Born distribution. In particular
\begin{align}
    \abs{\mel{x,y,z}{\mb{1}_n\otimes U_X(t)}{\psi_s}}^2=\cos(t/2)^2\abs{\braket{x,y,z}{\psi_s}}^2+\sin(t/2)^2\abs{\braket{x,\neg y,\neg z}{\psi_s}}^2=M_s^\eta(x,y,z)
\end{align}
\end{proof}

This leads us to the following concluding corollary:

\begin{corollary}[Local Free Fermion Distributions: Formal version of \Cref{inf:sample}]\label{cor:samp-hard}
Assuming LPN, then for any $d=n^{\Omega(1)}$ there is no efficient algorithm for learning $\M(n,d)$ with respect to an evaluator from samples.
In particular, there is no efficient algorithm for learning a free fermion state \(\rho\) with the unknown distribution \(P\in\mathcal{M}(n,d)\) as its output distribution.
\end{corollary}

Similarly to \Cref{cor:stat-hard-ff} we can make an even stronger statement in terms of non-local free fermion distributions.

\begin{corollary}[Non-local Free Fermion Distributions: Formal version of \Cref{inf:sample}]\label{cor:samp-hard-ff}
Assuming LPN there is no efficient algorithm for learning non-local free fermion distributions with respect to an evaluator from samples at any constant depth $d\geq2$.
In particular, there is no efficient algorithm for learning a free fermion state \(\rho\) with the unknown distribution as its output distribution.
\end{corollary}

The results from this section only apply to learning with respect to an evaluator. 
Note, however, that this implies corresponding statements for learning a free fermion state with a valid output distribution. 
To see this, recall that  amplitudes -- and hence an evaluator 
-- can  be efficiently calculated in the free-fermion formalism, which leads to the following corollary. 

\begin{corollary}[Implication for Learning the Underlying State]\label{cor:samp-hard-state}
    Assuming LPN there is no algorithm that, given access to samples in the occupation number basis of an unknown free fermion state \(\rho\), efficiently in $n,\delta^{-1}$ and \(\epsilon^{-1}\) learns with probability at least \(1-\delta\) a free fermion state \(\sigma\) such that the output distribution of \(\sigma\) is \(\epsilon\)-close in total variation distance to that of \(\rho\). 
\end{corollary}

Thus, even though our results do not exclude the possibility of efficiently learning a generator given samples from a free fermion distribution, the above line of reasoning strongly limits how the generator and an underlying free fermion state may be connected.
This limits the degree to which a learning algorithm could exploit the free fermion formalism.

%% file: content/conclusion.tex
\section{Conclusion and Discussion}

In this work we have shown that it is hard to learn free fermion, or likewise match gate distributions. 
we have shown in \Cref{cor:stat-hard,cor:stat-hard-ff} that exponentially many statistical queries are required to learn the unknown output distribution, and we have similarly shown in \Cref{cor:samp-hard,cor:samp-hard-ff,cor:samp-hard-state}, that there is no efficient algorithm for learning the unknown output distribution from individual samples.
In particular, learning the probability density -- or similarly any free fermion state with a valid output distribution -- from samples in the occupation number basis is at least as hard as LPN. A problem that is believed to be computationally hard.

Our work gives first results about the (non)-learnability of free fermion distributions. 
However, many questions remain open. 
Two immediate questions regard (1) the hardness of learning free fermion distributions with a fixed particle number as studied by \cite{aaronsonEfficientTomographyNonInteracting2023} and (2)  the average case hardness of free fermion distributions in the statistical query setting, similar to the analysis in \cite{nietner2023}.
We refer to Section 1.2.2 in \cite{aaronsonEfficientTomographyNonInteracting2023} for an instructive discussion of (1).
Regarding (2) we note that one can use the results by \cite{diazShowcasingBarrenPlateau2023} 
to calculate the statistical properties of the Haar measure over free fermion dynamics in analogy to \cite{nietner2023}. 
However, while \cite{diazShowcasingBarrenPlateau2023} show that many loss functions do admit a barren plateau they also show that many loss functions, in particular those with respect to local \(Z\)-observables, do not admit a barren plateau. 
This implies that their bounds are not sufficient to prove average case lower bounds with respect to the Haar measure over free fermion dynamics. 
To see this we note that exponential average-case lower bounds for statistical queries are equivalent to barren plateaus with respect to the average-case measure in all instances of the dual learning problem which, in case of learning the output distribution from statistical queries, corresponds to variationally optimizing diagonal observables (c.f. Section 8.2 in \cite{nietnerUnifyingQuantumStatistical2023}).

%% file: content/acknowledge.tex
I am thankful to many fruitfull discussions with Andreas Bauer, Marek Gluza, Marcel Hinsche, Marios Ioannu, Lennart Bittel, Ryan Sweke, Jonas Haferkamp and Jens Eisert. This work was supported by the BMBF (QPIC-1, Hybrid), DFG (CRC 183), the BMBK (EniQmA), and the Munich Quantum Valley (K-8).